\documentclass[lettersize,journal]{IEEEtran}

\usepackage{amsmath,amssymb}
\usepackage{subfigure}
\usepackage{hyperref}
\usepackage{graphicx,graphics,color,psfrag}
\usepackage{cite,balance}
\usepackage{caption}
\allowdisplaybreaks
\usepackage{algorithm}
\usepackage{algorithmic}
\usepackage{accents}
\usepackage{amsthm}
\usepackage{bm}
\usepackage{url}
\usepackage[english]{babel}
\usepackage{multirow}
\usepackage{enumerate}
\usepackage{cases}
\usepackage{stfloats}
\usepackage{dsfont}
\usepackage{color,soul}
\usepackage{amsfonts}
\usepackage{cite,graphicx,amsmath,amssymb}
\usepackage{subfigure}
\usepackage{fancyhdr}
\usepackage{hhline}
\usepackage{graphicx,graphics}
\usepackage{array,color}
\usepackage{mathtools}
\usepackage{amsmath}
\usepackage{amsthm}
\usepackage{caption}
\usepackage{circledsteps}

\include{header}
\newtheorem{theorem}{Theorem}

\newtheorem{corollary}{Corollary}

\usepackage{algorithm}
\def\BibTeX{{\rm B\kern-.05em{\sc i\kern-.025em b}\kern-.08em
		T\kern-.1667em\lower.7ex\hbox{E}\kern-.125emX}}
\begin{document}

	\title{Near-field Communications with Extremely Large-Scale Uniform Arc Arrays: Channel Modelling and Performance Analysis }
	
		\author{\IEEEauthorblockN{Guoyu Li,~\IEEEmembership{Student Member,~IEEE},~Changsheng You,~\IEEEmembership{Member,~IEEE},\\~Guanyu Shang, and~Shaochuan Wu,~\IEEEmembership{Senior Member, IEEE}}

		\IEEEcompsocitemizethanks{\IEEEcompsocthanksitem G. Li is with the Department of Electronic and Electrical Engineering, Southern University of Science and Technology (SUSTech), Shenzhen 518055, China, and also with the Electronic and Information Engineering, Harbin Institute of Technology (HIT), Harbin 150001, China (e-mail: lgy@stu.hit.edu.cn). 
		C. You is with the Department of Electronic and Electrical Engineering, Southern University of Science and Technology (SUSTech), Shenzhen 518055, China ( youcs@sustech.edu.cn).
		G. Shang is with the Innovation Photonics and Imaging Center, School of Instrumentation Science and Engineering, Harbin Institute of Technology (HIT), Harbin, China (email: shangguanyu95@163.com).
		S. Wu is with the School of Electronics and Information Engineering, Harbin Institute of Technology (HIT), Harbin 150080, China (e-mail: scwu@hit.edu.cn). (Corresponding authors: Changsheng You.)
		}
		% \IEEEcompsocitemizethanks{\IEEEcompsocthanksitem C. You is with the Department of Electronic and Electrical Engineering, Southern University of Science and Technology (SUSTech), Shenzhen 518055, China ( youcs@sustech.edu.cn).
		% }
		% \IEEEcompsocitemizethanks{\IEEEcompsocthanksitem G. Shang is with the Innovation Photonics and Imaging Center, School of Instrumentation Science and Engineering, Harbin Institute of Technology (HIT), Harbin, China (email: shangguanyu95@163.com).
		% }
		% \IEEEcompsocitemizethanks{\IEEEcompsocthanksitem S. Wu is with the School of Electronics and Information Engineering, Harbin Institute of Technology (HIT), Harbin 150080, China (e-mail: scwu@hit.edu.cn).
		% }
	}
	
	% \markboth{IEEE Wireless Communications Letters}%

	\maketitle
	
	\begin{abstract}
	
	% This letter studies a novel conformal array architecture, termed as extremely large-scale (XL) uniform arc array (UAA), which can seamlessly connect to curved structures at different angles, meeting the practical constraints of the installation structure. Considering the non-uniform spherical wave characteristics of the channel, we conducted mathematical modeling and performance analysis for XL-UAA. The results indicate that XL-UAA has a larger effective Rayleigh distance and uniform power distance than XL uniform linear array (ULA), and its closed-form signal-to-noise ratio (SNR) expression depends in a sophisticated manner on the collective properties of XL-UAA, including the distance between the user and the center of the circle and the radius of the arc. The numerical results show that XL-UAA has a better SNR than XL-ULA. Furthermore, we found that the SNR expression of XL-ULA is a special form of the SNR expression of XL-UAA. Asymptotic analysis demonstrates that the asymptotic SNR depends on the projection distance from the user to the top of the XL-UAA.
	% , which can flexibly attach to curved structures
	In this letter, we propose a new conformal array architecture, called \emph{extremely large-scale uniform arc array} (XL-UAA), to improve near-field communication performance. Specifically, under the non-uniform spherical wavefront channel model, we establish mathematical modeling and performance analysis for XL-UAAs. It is shown that XL-UAAs have larger direction-dependent Rayleigh distance and uniform power distance than the conventional XL uniform linear arrays (XL-ULAs). Moreover, a closed-form expression for the signal-to-noise ratio (SNR) is obtained, which depends on collective properties of XL-UAAs, such as the distance between the user and the array center, as well as the arc radius. In addition, we show that the asymptotic SNR of XL-UAAs with the number of antennas depends on the projection distance of the user to the middle of the arc array. Finally, numerical results verify that XL-UAAs achieve a higher SNR than XL-ULAs, especially at larger user incident angles.

	\end{abstract}
	
	\begin{IEEEkeywords}
		Extremely large-scale uniform arc array (XL-UAA), near-field communications, performance analysis.
	\end{IEEEkeywords}

	\section{Introduction}

	% \IEEEPARstart{T}{he} evolution of wireless communication networks has considerably contributed to the proliferation of the Internet of Things (IoT), and it is expected that the quantity of wireless IoT devices will escalate to 5000 billion by 2030. Along with the increasing automation and intelligence of IoT networks, there are two urgent issues that need to be addressed. On one hand, the data transmissions of vast IoT devices require a significant amount of power consumption. On the other hand, the ubiquitous connectivity of IoT devices will result in substantial occupation of spectrum resources, thereby exacerbating the problem of spectrum scarcity. The challenges posed by energy consumption and spectrum occupancy inevitably limit the future development of IoT networks, thus underscoring the need for enhancing energy-efficiency and spectrum utilization in IoT networks.

	The future sixth-generation communications (6G) impose more stringent performance requirements such as super-high spectral efficiency, ultra-reliable communications, extremely low latency, and massive connectivity, etc~\cite{you2024next}. As one of promising technologies for 6G, extremely large-scale arrays (XL-arrays) increase the number of antennas by at least an order-of-magnitude, which can significantly improve the spectral efficiency and spatial resolution of wireless systems~\cite{10496996}. Moreover, the drastically increased array aperture renders the users more likely to be located in the near-field region, resulting in fundamental changes in the electromagnetic propagation environment. In near-field communications, the conventional planar wavefront model is no longer accurate.Therefore, the more accurate spherical wavefront channel model needs to be considered, which brings both new opportunities and challenges~\cite{10239282, 9913211, zhou2024multi}.
    
	Among others, uniform linear arrays (ULAs) have been widely considered in the literature for near-field communications. For example, the authors in~\cite{lu2021does} analyzed the asymptotic performance of XL-ULAs. 
	The results showed that the resulting signal-to-noise ratio (SNR) under the non-uniform spherical wavefront (NUSW) model grows nonlinearly with the number of antennas according to a defined angular span, which demonstrates the importance of accurate modeling for near-field communications. 
	% In addition, XL-ULAs not only provide spatial resolution in the angular domain, but also endow additional spatial resolution in the distance domain. 
	% Specifically, the concept of location division multiple access was proposed, which demonstrated that using additional spatial resolution in the distance domain to provide services for users at different locations can significantly reduce inter-user interferences and improve spectral efficiency~\cite{10123941}. 
	Although the XL-ULA has shown significant advantages in near-field communications, it also faces several challenges. First, as the incident angle increases, the near-field region of XL-ULAs rapidly shrinks, which means that it can only provide near-field services for areas within small incident angles~\cite{10243590}. Second, XL-ULAs impose more stringent requirements for the installation environment which may not be suitable for installation on curved objects~\cite{10485481}. 
	In order to overcome these limitations, \emph{conformal} arrays have emerged as a promising solution. 
    Specifically, some preliminary studies have explored uniform circular arrays (UCAs) to extend the near-field region within a 360° angular space, hence allowing more users to enjoy the benefits of near-field communications~\cite{10243590}. 
	Moreover, efficient beam training schemes suitable for narrowband and wideband XL-UCA systems were proposed in~\cite{10005200} and~\cite{chen2024near} to effectively reducing the required beam training overhead.
    However, it is impossible for XL-UCAs to serve users through all radiating antennas. In practice, for hotspots in specific regions, only portions of XL-UCA antennas facing the users contribute to useful received signal power.

	To address the above issues, we consider in this letter a new conformal array, called XL uniform arc arrays (XL-UAAs). As illustrated in Fig.~\ref{System_model}, XL-UAAs can more efficiently utilize the available area in limited space. Moreover, XL-UAAs have appealing flexibility and adaptability, which can be flexibly attached to the curved structures such as street lights~\cite{10485481}. 
	To characterize the communication performance of XL-UAAs, we first establish the mathematical modeling and analyze the near-field communication performance of XL-UAAs under the general NUSW channel model. Specifically, the direction-dependent Rayleigh (DDRayl) distance and uniform power distance (UPD) of the XL-UAA are first obtained, which show that XL-UAAs have larger DDRayl distance and UPD than XL-ULAs. 
	Subsequently, a closed-form expression for the received SNR at the user in XL-UAA systems is obtained, which depends on collective properties of the XL-UAA in a sophisticated manner. The analysis reveals that the developed channel modelling for XL-UAAs includes that of XL-ULAs as a special case. 
	Finally, we analytically show that the asymptotic SNR of XL-UAAs depends on the projection distance of the user to the middle of the arc array.
	% Motived by the above, we study the mathematical modeling and performance analysis for XL-UAAs. Under the NUSW characteristics, the direction-dependent Rayleigh (DDRayl) distance and uniform power distance (UPD) of the XL-UAA are first derived, and the results show that XL-UAA has larger DDRayl distance and UPD than XL-ULA. Subsequently, a closed-form SNR expression of the maximum ratio combining (MRC) beamforming of XL-UAA is derived, which shows that the obtained closed-form SNR expression depends on the collective properties of the XL-UAA in a sophisticated manner. The analysis reveals that the developed modeling of XL-UAA is a general form of XL-ULA. Asymptotic analysis shows that as the number of antennas increases, the SNR gradually approaches a constant, which depends on the projection distance from the user to the top of XL-UAA.

	\section{System Model}
	
	\begin{figure}
		\centering
		\includegraphics[width = 0.7\linewidth]{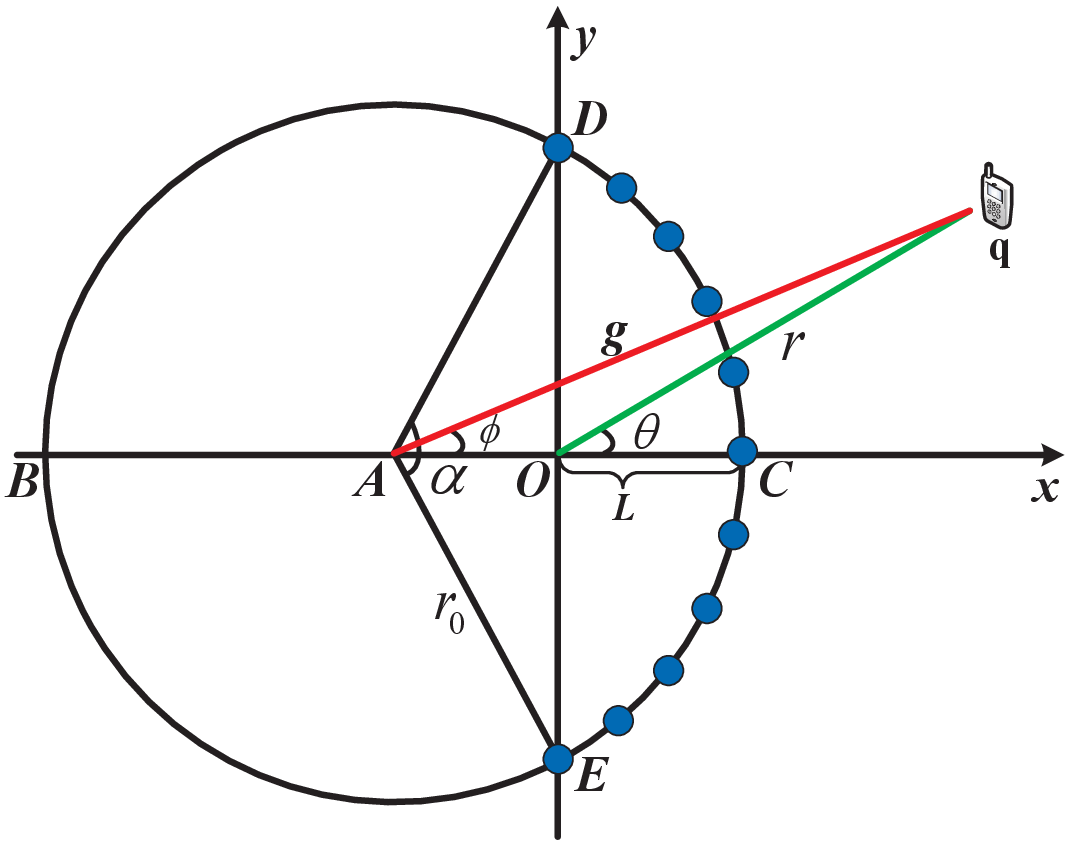}
		\caption{The geometrical relationship between the XL-UAA and user in the near-field region.}
		\label{System_model} 
		\vspace{-5pt}
	\end{figure}
	We consider an XL-UAA communication system as shown in Fig.~\ref{System_model}, where a base station (BS) equipped with an $M$-antenna XL-UAA communicates with a single-antenna user\footnote{The obtained results can be extended to the multi-user case by considering more complicated beamforming techniques at the BS similar as in \cite{zhou2024multi}, such as zero-forcing, and minimum mean-square error.}.

	\underline{\textbf{Near-field Channel Model:}} 
	We assume that the user is located in the near-field region of the XL-UAA, where the BS-user distance is smaller than the (effective) Rayleigh distance for the XL-UAA, which will be presented in Section~\ref{performance_analysis}. As such, we consider the near-field spherical wavefront channel modeling for XL-UAAs, as detailed below.

	Without loss of generality, we assume that $M$ is an odd number. 
	As a general form of XL-UCAs, the XL-UAA connecting the arc between $D$ and $E$ is placed on the $x$-$y$ plane, where the central antenna is located at the middle of arc, i.e., $C$ in Fig.~\ref{System_model}.
	The inter-antenna spacing $d$ is set as half-wavelength, i.e., $d = \lambda/2$ with $\lambda$ denoting the carrier wavelength. To facilitate the analysis, we let $A$ denote the center of the arc, $r_0$ denote the arc radius, $L$ denote the length of line segment connecting the origin $O$ and the point $C$, and $\alpha$ denote the central angle formed by the line segments of $AD$ and $AE$. Then, the inter-antenna angle relative to the arc center $A$ is $\alpha_0 = \frac{\alpha}{M-1}$. 
	Furthermore, we consider the case $L \le r_{0}$, so that all antennas in the XL-UAA can provide services to users.
	From the geometry relationship, the location of antenna $m$ is given by $\mathbf{w}_m = \left[ r_0 \cos \alpha_{m}  - (r_0 -L), r_0 \sin \alpha_{m} \right]^{T}$, where $\alpha_m = m \alpha_0$, $m = 0,\pm 1,\cdots,\pm (M-1)/2$. 
	
	Let $\mathbf{q} = \left[ r \cos\theta, r \sin\theta \right]^{T}$ denote the location of the user, where $r$ represents the distance between the origin $O$ and the user, and $\theta \in \left[ -\frac{\pi}{2},\frac{\pi}{2} \right]$ is the angle of the user relative to the origin, i.e., the angle spanned by two line segments $OP$ and $OC$. Then, the distance between the $m$th antenna and the user can be expressed as 
	\begingroup\makeatletter\def\f@size{9}\check@mathfonts
		\begin{equation}
			\begin{aligned} \label{equ_2}
				&r_{m} = \left\| \mathbf{w}_m - \mathbf{q} \right\|    \\
				& = \sqrt{g^{2} + r_{0}^{2} - 2{r_0}r \sin\theta \sin\alpha_m - 2{r_0}(r \cos\theta + ({r_0}-L)) \cos\alpha_m} \\
				& = \sqrt{g^{2} + r_{0}^{2} - 2{r_0}g \sin\phi \sin({m\varepsilon}) - 2{r_0}g \cos\phi \cos({m\varepsilon})},
			\end{aligned} 
		\end{equation}
	\endgroup 
	where $g = \sqrt{{r^2} + (r_0-L)^2 + 2r(r_0-L)\cos\theta}$ is the distance between the user and the arc center $A$, $\phi$ is the angle spanned by the line segments $AP$ and $AC$, and $\varepsilon \buildrel \Delta \over = \alpha_0 = \frac{\alpha}{M-1}$. Furthermore, we consider the case where the user is outside the arc, so that the distance from the user to the center of the arc (i.e., node A) is greater than the radius of the arc, i.e., $g>r_0$. Then, based on the NUSW channel model, the general near-field multi-path channel from the user to the BS can be modeled as
	\begingroup\makeatletter\def\f@size{10}\check@mathfonts
		\begin{equation} \label{Channel_model}
		  \mathbf{h}_{\mathrm{near}} = \mathbf{b}\left( r,\theta \right) + \sum_{\ell = 1}^{L} \mathbf{b}_{\ell}\left( r_{\ell},\theta_{\ell} \right),
		\end{equation}
	\endgroup
	which includes one LoS path and $L$ NLoS paths. Here, the parameters $r_{\ell}$ and $\theta_{\ell}$ are the distance and spatial angle of the $\ell$th path, respectively. Similar to~\cite{zhou2024multi}, we mainly considers the near-field communication scenarios in high-frequency bands such as mmWave and THz. As such, the NLoS channel paths in these scenarios exhibit negligible power due to severe path-loss and shadowing. Then the channel $\mathbf{h}_{\mathrm{near}}$ can be approxiamted as $\mathbf{h}_{\mathrm{near}} = \mathbf{b}\left( r,\theta \right)$, where $\mathbf{b}\left( r,\theta \right)$ is given by
	\begingroup\makeatletter\def\f@size{10}\check@mathfonts
		\begin{equation}
			\left[ \mathbf{b}\left(r,\theta \right) \right]_{m} = \frac{\sqrt{\beta_0} }{r_{m}}{e^{-j\frac{2\pi}{\lambda}{r_{m}}}}, \forall m = 0, \pm 1, \cdots, \pm \frac{M-1}{2}
		\end{equation}
	\endgroup
	with $\beta_0$ representing the reference channel gain at the distance of 1 m.
	% where ${{\rm{b}}_m}\left( {r,\theta} \right) = \frac{1}{r_{m}}{e^{-j\frac{2\pi}{\lambda}{r_{m}}}}$ is the antenna-wise complex-valued channel gain.
	  
	\underline{\textbf{Signal Model:}} 
	We consider XL-UAA uplink communication in this letter, while the obtained results can be easily extended to the downlink case. Let $s \in \mathbb{C}$ denote the transmitted symbol by the user with power $P$ and $\mathbf{v} \in \mathbb{C}^{M \times 1}$ represent the linear receive beamforming vector with $\left\| \mathbf{v} \right\| = 1$. The resulting signal after the receive beamforming is given by
	\begingroup\makeatletter\def\f@size{10}\check@mathfonts
		\begin{equation}
		  y = \sqrt{P} \mathbf{v}^{H} {\bf{h}_{\mathrm{near}}} s + \mathbf{v}^{H} \mathbf{n}, 
		\end{equation}
	\endgroup
	where $\mathbf{n} \sim \mathcal{N}_{\mathbb{C}}\left(\mathbf{0}_{M},\sigma^{2}\mathbf{I}_{M}\right)$ is the received additive white Gaussian noise (AWGN) vector. Then, the received SNR at the user is given by
	\begingroup\makeatletter\def\f@size{10}\check@mathfonts
		\begin{equation}
		  \gamma_{\mathrm{UAA}} = \frac{P \left| \mathbf{v}^{H} {\bf{h}_{\mathrm{near}}} \right|^2}{ \left\| \mathbf{v}^{2} \right\| \sigma^{2}} = \bar{\gamma}_0 \left| \mathbf{v}^{H} \mathbf{b}\left(r,\theta \right) \right|^2, 
		\end{equation}
	\endgroup
	where $\bar{\gamma}_0 = \left( P \beta_0 \right)/\sigma^{2}$ is the received SNR at the reference distance.

	\section{Near-field Communication Performance for XL-UAAs} \label{performance_analysis}

	In this section, we first present the Rayleigh distance for the XL-UAA and then characterize its SNR performance in closed form based on maximum ratio combining (MRC). Subsequently, we provide asymptotic SNR performance analysis for XL-UAAs when the number of antennas is sufficiently large. Useful insights are drawn to show the performance gains of XL-UAAs as compared to conventional XL-ULAs.

	\subsection{Rayleigh distance for XL-UAAs}

	Under the NUSW model, both the linear phase approximation and the uniform amplitude approximation are invalid. Thus, we adopt the DDRayl distance and UPD to show the impact of signal direction on the phase variations and the amplitude variations across the XL-UAA antennas, respectively. 

	\subsubsection{The DDRayl distance} 

	The DDRayl distance takes into account the influence of signal direction on the phase variations across the array antennas. Let $\Delta(r,\theta)$ define the maximum phase error across the array antennas in the distance $r$ and direction $\theta$, which can be expressed as~\cite{10496996}
	% which is the difference between the exact phase and that based on the first-order Taylor phase approximation, expressed as
	\begingroup\makeatletter\def\f@size{10}\check@mathfonts
		\begin{equation} \label{Phase_difference}
			\Delta(r,\theta) \buildrel \Delta \over = \max_{m} \frac{2\pi}{\lambda}\left(r_m - r_m^{\mathrm{first}}\right), 
		\end{equation}
	\endgroup	
	where $r_m^{\mathrm{first}} = r_m + (r_0-L)\cos\theta-r_0\cos(\theta-\alpha_m)$ is the first-order Taylor distance approximation of $r_m$. According to~\cite{lu2021communicating}, the DDRayl distance can be obtained by imposing the maximum phase variation $\pi/8$ as follow: 
	\begingroup\makeatletter\def\f@size{10}\check@mathfonts
  		\begin{equation} \label{DDRayl}
    		r_{\mathrm{Ray1}}^{\mathrm{UAA}}(\theta) \buildrel \Delta \over = \mathrm{arg} \min_{r} \Delta(r,\theta) \le \frac{\pi}{8}.
  		\end{equation}
	\endgroup
	% begingroup\makeatletter\def\f@size{10}\check@mathfonts
  	% 	\begin{equation} \label{DDRayl}
    % 		r_{\mathrm{Ray1}}^{\mathrm{UAA}}(\theta) \buildrel \Delta \over = \mathrm{arg} \min_{r} \Delta(r,\theta) \le \frac{\pi}{8}.
  	% 	\end{equation}
	% \endgroup
	By replacing the exact distance $r_m$ in~\eqref{Phase_difference} with its second-order Taylor approximation, the DDRayl distance of XL-UAA can be approximately obtained as 
	\begingroup\makeatletter\def\f@size{9}\check@mathfonts
  		\begin{equation} \label{DDRayl_approx}
    		r_{\mathrm{Ray1}}^{\mathrm{UAA}}(\theta) \approx \max_{m} \left\{ \frac{ 8[ ({r_0} - L)\sin\theta - {r_0}\sin(\theta-\alpha_{m}) ]^2}{\lambda}\right\} . 
  		\end{equation}
	\endgroup
	Due to the sophisticated structure of the XL-UAA, it is difficult to obtain a closed-form expression for the DDRayl distance, while the explicit value in \eqref{DDRayl_approx} can be obtained numerically. For example, given the same aperture, the DDRayl distance of the XL-UAA at $\theta = 0$ is $r_{\mathrm{Ray1}}^{\mathrm{UAA}}(0) = 8({r_0}\sin\frac{\alpha}{2})^2/\lambda$, which is equal to the DDRayl distance of the XL-ULA at $\theta = 0$. However, the DDRayl distance of the XL-UAA at $\theta = \frac{\pi}{2}$ is $r_{\mathrm{Ray1}}^{\mathrm{UAA}}(\frac{\pi}{2}) = \frac{8L^2}{\lambda}$, which is much larger than the DDRayl distance of XL-ULA $r_{\mathrm{Ray1}}^{\mathrm{ULA}}(\frac{\pi}{2}) = 0$ at $\theta = \frac{\pi}{2}$ given the same aperture. This thus shows that XL-UAAs can provide a larger effective near-field region than XL-ULAs when the users are located in the large incident angles.

	% Specifically, the DDRayl distance of the XL-UAA is equal to the DDRayl distance of the XL-ULA with the same aperture at $\theta = 0$, i.e., $r_{\mathrm{DDRayl}}^{\mathrm{UAA}}(0) =  8({r_0}\sin\frac{\alpha}{2})^2/\lambda = 2d_{DE}^2/{\lambda}$, where $d_{DE}$ is the array aperture. The DDRayl distance of the XL-UAA is greater than the DDRayl distance of XL-ULA at $\theta = \frac{\pi}{2}$, i.e., $r_{\mathrm{DDRayl}}^{\mathrm{UAA}}(\frac{\pi}{2}) = \frac{8L^2}{\lambda} > r_{\mathrm{DDRayl}}^{\mathrm{ULA}}(\frac{\pi}{2}) = 0$.

	\subsubsection{UPD} 
	Unlike the (direction-dependent) Rayleigh distance, UPD focuses on the amplitude variations across the XL-UAA. Let $\Upsilon(r,\theta)$ denote the ratio of the weakest and the strongest power across the array antennas in the distance $r$ and direction $\theta$. Based on the geometry relationship, $\Upsilon(r,\theta)$ can be obtained as~\cite{10496996}
	\begingroup\makeatletter\def\f@size{10}\check@mathfonts
		\begin{equation} \label{UPD__approx}
			\Upsilon(r,\theta) = \frac{\min\limits_{m} \left\{ U(m) \right\}}{r^2+r_0^2\sin^2\frac{\alpha}{2}+2rr_0\sin\theta\sin\frac{\alpha}{2}} , 
		\end{equation}
	\endgroup	
    where $U(m) = r^2+r_0^2+(r_0-L)^2-2rr_0\cos(\theta-\alpha_m)-2(r_0-L)(r_0\cos\alpha_{m}-r\cos\theta)$. According to~\cite{10496996}, UPD can be obtained by imposing a minimum threshold for $\Upsilon(r,\theta)$:
	\begingroup\makeatletter\def\f@size{10}\check@mathfonts
		\begin{equation} \label{UPD}
			r_{\mathrm{Ray2}}^{\mathrm{UAA}}(\theta) \buildrel \Delta \over = 
			\mathrm{arg} \min_{r} \Upsilon(r,\theta) \ge \Upsilon_{\mathrm{th}}, 
		\end{equation}
	\endgroup	
	where $\Upsilon_{\mathrm{th}}$ is a certain threshold. Similar to the DDRayl distance, it is also difficult to obtain a closed-form expression for the UPD of XL-UAA. However, several useful insights can be obtainted. For example, the power ratio of the XL-UAA is equal to that of the XL-ULA at $\theta = \frac{\pi}{2}$ given the same aperture. Therefore, the UPDs of XL-UAA and XL-ULA are also the same at $\theta = \frac{\pi}{2}$. In addition, when $\theta = 0$ and $\Upsilon_{\mathrm{th}}=0.9$, the UPD of XL-UAA $r_{\mathrm{Ray2}}^{\mathrm{UAA}}(0)=(20L+\sqrt{390L^2+36D^2})/2$ is greater than that of XL-UAA $r_{\mathrm{Ray2}}^{\mathrm{ULA}}(0)=3D$, where $D$ is the array aperture.
	
	\subsection{SNR analysis for XL-UAA}

	It is well known that for single-user communication, MRC (i.e., $\mathbf{v} = \frac{{\bf{b}}\left( {r,\theta} \right)}{\left \| {\bf{b}}\left( {r,\theta} \right) \right \| }$) is the optimal receive beamforming. As such, the resulting maximum SNR, $\gamma_{\mathrm{UAA}} = \bar{\gamma}_0 \left \| {\bf{b}}\left( {r,\theta} \right) \right \|^{2}$, is given by
	\begingroup\makeatletter\def\f@size{10}\check@mathfonts
  		\begin{equation}  \label{sum_SNR}
    		\begin{aligned}
				\gamma&_{\mathrm{UAA}} = \bar{\gamma}_0 \left \| {\bf{b}}\left( {r,\theta} \right) \right \|^{2} = \bar{\gamma}_0 \sum \limits_{m = - \frac{{M-1}}{2}}^{\frac{{M-1}}{2}} \frac{1}{r_m^2} = \sum\limits_{m = -\frac{{M-1}}{2}}^{\frac{{M-1}}{2}} \\
      			& \times \frac{\bar{\gamma}_0}{g^{2} + r_{0}^{2} - 2{r_0}g \cos\phi \cos({m\varepsilon}) - 2{r_0}g \sin\phi \sin({m\varepsilon}) }.
    		\end{aligned} 
  		\end{equation}
	\endgroup

	\begin{figure}
		\centering
		\includegraphics[width=0.8\linewidth]{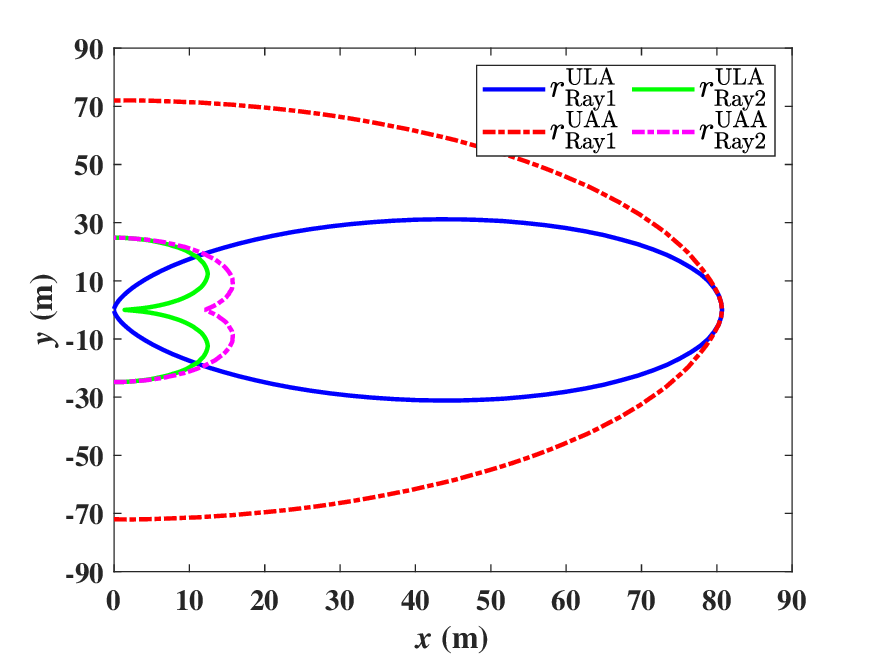}
		\caption{Comparison of the DDRayl distance and UPD.}
		\label{Distances}
		\vspace{-5pt}
	\end{figure}
	
	\begin{theorem} \label{Theorem1}
		Given the MRC beamforming, the received SNR at the XL-UAA in \eqref{sum_SNR} can be expressed as
		\begingroup\makeatletter\def\f@size{10}\check@mathfonts
			\begin{equation} \label{max_SNR}
				\gamma_{\mathrm{UAA}} = 2 \frac{\bar{\gamma}_0 \left( {M-1} \right)}{\alpha ({g^2}-{r_0^2})} \times U(g,r_0), 
			\end{equation}
		\endgroup
		where $U(x,y) = \arctan \left( \frac{(x^2 + y^2 + 2xy\cos{\phi}) \sqrt{\frac{L}{2y-L}}-2xy\sin{\phi}}{x^2 - y^2}\right) + \arctan \left( \frac{(x^2 + y^2 + 2xy\cos{\phi}) \sqrt{\frac{L}{2y-L}}+2xy\sin{\phi}}{x^2 - y^2}\right)$.
	\end{theorem}
	
	\begin{proof}
		Please refer to Appendix A.
	\end{proof}
	  
	\textbf{Theorem~\ref{Theorem1}} shows that with the more generic spherical wavefront XL-UAA channel model in~\eqref{Channel_model}, the resulting SNR of the UAA \emph{no longer} scales linearly with the antenna number $M$ as in the XL-ULA case~\cite{lu2021does}. Instead, it depends on collective properties of the XL-UAA in a more intricate manner, including the distance between the user and the arc center  $g$, and the arc radius $r_0$. To gain further insights, some special cases are considered in the following.

	\begin{corollary} \label{Corollary1}
		When the arc support $L$ $\rightarrow$ 0, the resulting SNR $\gamma_{\mathrm{UAA}}$ in \eqref{max_SNR} reduces to
		\begingroup\makeatletter\def\f@size{10}\check@mathfonts
			\begin{equation}
				\begin{aligned}
					\gamma_{\mathrm{UAA}} =  & \frac{\bar{\gamma}_0}{d r \cos\theta}\left[ \arctan \left(\frac{M d}{2 r \cos\theta}-\tan \theta\right) \right. \\
					&\left. +\arctan \left(\frac{M d}{2 r \cos \theta}+\tan \theta\right)\right], 
				\end{aligned}
			\end{equation}
		\endgroup
		which is consistent with the result for XL-ULAs in \cite{lu2021does}. 
	\end{corollary}

	\begin{proof}
		Please refer to Appendix B.
	\end{proof}

	\textbf{Corollary~\ref{Corollary1}} shows that when the arc support $L \rightarrow 0$, the SNR expression of the XL-UAA reduces to that of the XL-ULA. This shows that the new closed-form SNR expression of XL-UAA generalizes the closed-form SNR expression of XL-ULA, which applies to both the XL-UAA and XL-ULA.

	\begin{figure*}[t]
		\centering
		\subfigure[SNRs versus the array aperture $D$.]{\label{Asymptotic_SNR}
			\includegraphics[height=4.23cm]{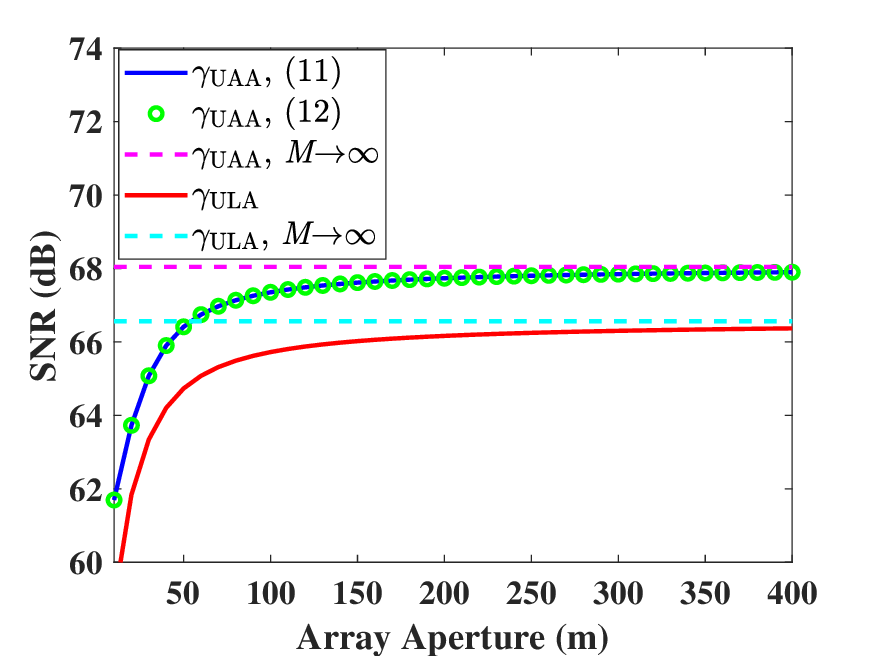}}
		\hspace{5pt}
		\subfigure[SNRs versus the arc support distance $L$.]{\label{L_change}
			\includegraphics[height=4.2cm]{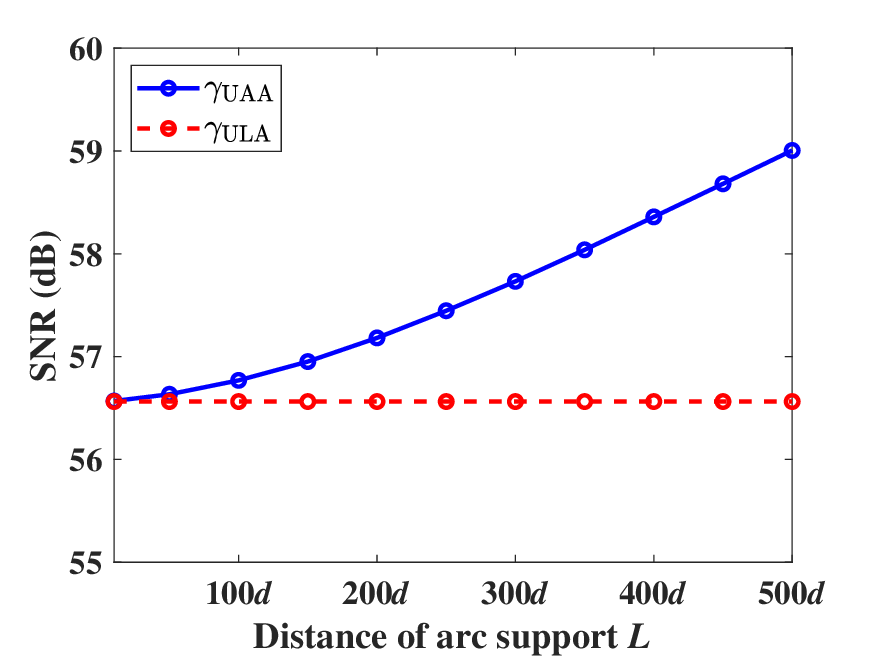}}
		\hspace{5pt}
		\subfigure[SNRs versus user angle $\theta$.]{\label{theta_change}
			\includegraphics[height=4.2cm]{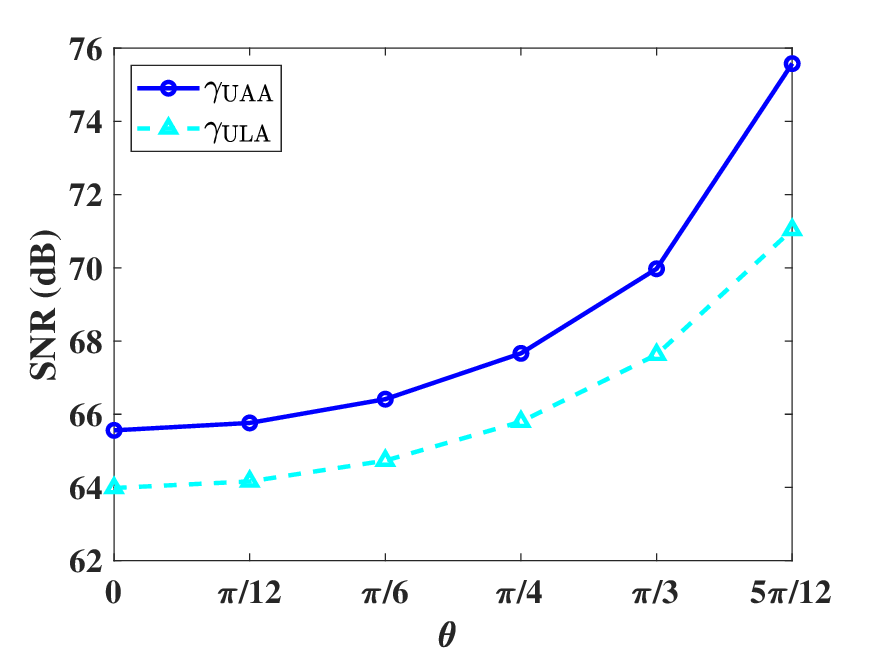}}
		\vspace{-5pt}
		\caption{The SNRs for XL-UAA and XL-ULA.}
		\label{SNRs}
		
	\end{figure*} 
    % \begin{figure}
	% 	\centering
	% 	\includegraphics[width=0.9\linewidth]{Asymptotic_SNR.pdf}
	% 	\caption{The SNRs versus the array aperture for XL-UAA and XL-ULA.}
	% 	\label{Asymptotic_SNR}
	% \end{figure}

	\begin{corollary}
		When the number of XL-UAA antennas goes to infinity, the asymptotical received SNR at the XL-UAA with the MRC receive beamforming is
		\begingroup\makeatletter\def\f@size{10}\check@mathfonts
			\begin{equation} \label{M_infinity}
				\gamma_{\mathrm{UAA}} =  \frac{\bar{\gamma}_0 \pi}{{d(r \cos\theta - L)}}. 
			\end{equation}
		\endgroup
	\end{corollary}

	\begin{proof}
		Please refer to Appendix C.
	\end{proof}

	It can be shown from~\eqref{M_infinity} that for the XL-UAA with an extremely large number of antennas, the obtained asympotic SNR approaches a constant value determined by the projected distance of the user to the middle of the arc, i.e., $r \cos\theta - L$. Compared with the asymptotic SNR for XL-ULAs that depends only on the projection distance from the user to the array, i.e., $r \cos\theta$~\cite{lu2021does}, the asymptotic SNR of an XL-UAA with $N \to \infty$ has a larger value. This can be intuitively understood, since in the near-field region, the antennas close to the user attribute to more dominant power, while the ones sufficiently farther away from the user provide little SNR enhancement.  
	% This is because the XL-UAA is closer to the user than the XL-ULA with the same aperture. 

	\section{Numerical Results}

	In this section, numerical results are provided to verify our theoretical analysis for XL-UAA near-field communications. We consider an XL-UAA BS operating at 30 GHz which has half-wavelength antenna spacing. The received SNR at reference distance is set as $\bar{\gamma}_0 = 50$ dB~\cite{lu2021does}. For fair comparison, we consider that both the XL-UAA and XL-ULA have the same array aperture~\cite{10243590}.

    % \begin{figure}
	% 	\centering
	% 	\includegraphics[width=0.9\linewidth]{L_change.pdf}
	% 	\caption{The SNRs versus the distance $L$ for XL-UAA and XL-ULA.}
	% 	\label{L_change}
	% \end{figure}

	The DDRayl distance and UPD of XL-ULAs and XL-UAAs are plotted in Fig.~\ref{Distances}, where the array aperture is 0.635 m and the arc support is $L$ = 60 $d$. It is observed that the DDRayl distance and UPD of XL-UAAs are much larger than those of XL-ULAs at all angles. Compared with XL-ULAs, XL-UAAs provide larger DDRayl distance and UPD in the directions of large incidence angles and small incidence angles, respectively. Therefore, XL-UAAs can provide a larger near-field region than XL-ULAs, allowing more users to be located in the near-field of the array.

	Fig.~\ref{Asymptotic_SNR} shows the SNRs for the XL-UAA and XL-ULA systems versus the array aperture. The distance between the user and the origin is $r = 16$ m, the arc support is $L = 800$ $d$, and the angle of the line segment connecting to the user location and the origin relative to the $x$-axis is $\theta = 30^{\circ}$. For XL-UAAs, the SNR results based on both the summation in~\eqref{sum_SNR} and the closed-form expression~\eqref{max_SNR} are shown, together with the SNR limit in~\eqref{M_infinity}. It is observed from Fig.~\ref{Asymptotic_SNR} that the closed-formed SNR expression in~\eqref{max_SNR} perfectly matches with the expression in~\eqref{sum_SNR}, which verifies \textbf{Theorem~\ref{Theorem1}}. Moreover, the results of the XL-UAA also approach to a constant, as the number of antennas increases. This demonstrates the importance of properly modelling the spherical wave in XL-UAA near-field communications. Furthermore, it is worth noting that the SNR for XL-UAAs is higer than that for XL-ULAs. This is expected because the XL-UAA is closer to the user and can deploy more antennas than the XL-UAA at the same array aperture.

	Fig.~\ref{L_change} shows the relationship between SNRs and the arc support $L$. The distance between the user and the origin is $r=16$ m, the angle of the line segment connecting to the user location and the origin relative to the $x$-axis is $\theta=30^{\circ}$, and the array aperture is $5$ m. It is observed that as the arc support $L$ increases, the SNR of XL-UAA increases, which shows that the larger the arc support $L$ is, the greater the SNR of XL-UAA is. Furthermore, when the arc support $L$ is very small, the SNR of XL-UAA is equal to that of XL-ULA, which also verifies the correctness of \textbf{Corollary~\ref{Corollary1}}.

	Fig.~\ref{theta_change} shows the relationship between SNRs and the user angle $\theta$. The distance between the user and the origin is $r=16$ m, the arc support is $L=800$ $d$, and the array aperture is $50$ m. It is observed that as the angle $\theta$ increase, the SNRs of both XL-UAA and XL-ULA increases, and the difference in the SNRs between XL-UAAs and XL-ULAs becomes larger. This shows that the larger the incident angles, the better the performance of XL-UAAs compared to XL-ULAs.

	\vspace{-5pt}
	\section{Conclutions}
	
	In this letter, we investigated the mathematical modeling and performance analysis for the newly-considered XL-UAAs. Under the NUSW channel model, the DDRayl distance and UPD of XL-UAA were obtained, which shows that XL-UAAs have larger DDRayl distance and UPD than XL-ULAs. A closed-form SNR expression for XL-UAA was then derived, indicating that the SNR depends on the collective properties of XL-UAAs. The analysis revealed that XL-UAA modeling is a general form of XL-ULA. Moreover, asymptotic analysis demonstrated that the asymptotic SNR depends on the projection distance from the user to the middle of the arc array. 
	% The numerical results show that XL-UAA has a better SNR than XL-ULA, especially at large incident angles.

	% In this letter, we study the mathematical modeling and performance analysis of the XL-UAA. Considering the NUSW characteristics, the direction-dependent Rayleigh (DDRayl) distance and uniform power distance (UPD) of the XL-UAA are first derived, and the results show that XL-UAA has larger DDRayl distance and UPD than XL-ULA. Subsequently, a closed-form SNR expression of the maximum ratio combining (MRC) beamforming for a single user of XL-UAA is derived. The results show that the obtained closed-form SNR expression depends on the collective properties of the XL-UAA in a sophisticated manner.  The analysis reveals that the developed modeling and performance analysis of XL-UAA is a general form of XL-ULA. Furthermore, asymptotic analysis shows that as the number of antennas increases, the SNR gradually approaches a constant, which depends on the projection distance from the user to the top of XL-UAA.
	% Furthermore, its closed-form SNR expression depends on the collective properties of UAA in a sophisticated manner, and the closed-form SNR expression of XL-ULA is a special form of the XL-UAA SNR expression. The numerical results show that XL-UAA has a better SNR than XL-ULA. 
	% In addition, asymptotic analysis indicates that the asymptotic SNR depends on the projection distance from the user to the top of the XL-UAA.

	\appendices

	\section{Proff of Theorem 1}

	Based on the SNR expression~\eqref{max_SNR}, we first define the function $f\left( x \right) \buildrel \Delta \over = \frac{1}{{a + b \cos{x} + c \sin{x}}}$ over the interval $x \in \left[ { - \frac{M}{2}\varepsilon ,\frac{M}{2}\varepsilon } \right]$. The interval is partitioned into $M$ subintervals each of equal length $\varepsilon$. 
	Since $\varepsilon \ll 1$ in practice ($M$ is very large), we have $f\left( x \right) \approx f\left( m\varepsilon \right)$, $\forall x\in \left[ {(m-\frac{1}{2})\varepsilon ,(m+\frac{1}{2})\varepsilon } \right]$. Then, based on the concept of definite integral, we have
	\begingroup\makeatletter\def\f@size{9}\check@mathfonts
		\begin{equation} \label{Integral}
			\sum_{-\frac{M-1}{2}}^{\frac{M-1}{2}} f\left( m\varepsilon \right) \varepsilon \approx \int_{-\frac{M}{2} \varepsilon}^{\frac{M}{2}\varepsilon} f\left( x \right) dx, 
		\end{equation}  
	\endgroup
	where~\eqref{Integral} holds since $\varepsilon \ll 1$. By substituting $f(x)$ into~\eqref{Integral}, we have~\eqref{Appendix_A}, shown at the top of the this page, where (a) follows from the integral formula 2.558 in~\cite{gradshteyn2014table}, i.e., $\int \frac{{dx}}{{a + b \cos{x} + c \sin{x}}} = \frac{2}{\sqrt{a^2 - b^2 - c^2}} {\arctan\left( {\frac{{\left( {a-b} \right) \tan{x}+c}}{\sqrt{a^2 - b^2 - c^2}}} \right)}$ for $a^2 > b^2 + c^2$. Thus, the proof of Theorem 1 is thus completed.
	
	\begin{figure*} [ht]%hb代表放在文章底部，%ht为放在文章顶部 
		\centering
		\begingroup\makeatletter\def\f@size{7.5}\check@mathfonts
			\begin{equation}  \label{Appendix_A}
				\begin{aligned}
					& \sum_{-\frac{M-1}{2}}^{\frac{M-1}{2}} \frac{1}{{a + b \cos\left( {m \varepsilon} \right) + c \sin\left( m \varepsilon \right)}}  \approx \frac{1}{\varepsilon} \int_{-\frac{M}{2} \varepsilon}^{\frac{M}{2}\varepsilon} \frac{1}{{a + b \cos{x} + c \sin{x} }} dx {\buildrel (a) \over =} \left. {\frac{1}{\varepsilon}\frac{2}{\sqrt{a^2 - b^2 - c^2}}\arctan \left( {\frac{{\left( {a - b} \right)\tan {\frac{x}{2}} + c}}{\sqrt{a^2 - b^2 - c^2}}} \right)} \right|_{-\frac{M}{2}\varepsilon}^{\frac{M}{2}\varepsilon} \\
					% & = \frac{2 \left( M-1 \right)}{\alpha \left| g^2 - r_0^2 \right|} \left. \arctan \left( \frac{\left( g^2 + r_0^2 + 2 g r_0 \cos\phi \right) \tan{\frac{x}{2}} -2 g {r_0} \sin\phi}{\left| g^2 - r_0^2 \right|} \right) \right|_{-\frac{M}{2}\varepsilon}^{\frac{M}{2}\varepsilon}\\
					& = \frac{2 \left( M-1 \right)}{\alpha (g^2 - r_0^2)} \left[ {\arctan\left( \frac{\left( g^2 + r_0^2 + 2 g r_0 \cos\phi \right) \sqrt{\frac{L}{2r_0 - L}} -2 g {r_0} \sin\phi}{g^2 - r_0^2} \right) + \arctan\left( \frac{\left( g^2 + r_0^2 + 2 g r_0 \cos\phi \right) \sqrt{\frac{L}{2r_0 - L}} +2 g {r_0} \sin\phi}{g^2 - r_0^2} \right)} \right].
				\end{aligned} 
			\end{equation}
		\endgroup
	\end{figure*}
    \vspace{-10pt}
	\begin{figure*}[ht] %hb代表放在文章底部，%ht为放在文章顶部 
		\centering
		\begingroup\makeatletter\def\f@size{7.5}\check@mathfonts
		  \begin{equation}  \label{M_inf}
			\begin{aligned}
			  &\gamma_{\mathrm{UAA}} = \frac{2\bar{\gamma}_0 (M-1)}{\alpha ({{r^2}+{L^2}-2rL\cos\theta+2{r_0}(r\cos\theta-L)})} \left.\arctan \left(\frac{\left.\left(4{r_0}^2+4{r_0}(r \cos\theta-L)+{r^2}+L^2-2rL\cos\theta\right) \tan{\frac{\alpha}{4}} - 2 r r_{0} \sin\theta \right)}{({r^2}+{L^2}-2rL\cos\theta+2{r_0}(r\cos\theta-L))}\right)\right|_{-\frac{M}{2} \varepsilon} ^{\frac{M}{2} \varepsilon} \buildrel (a) \over = \\
			  & \mathop {\lim}\limits_{{r_0} \to \infty} \frac{{2\bar{\gamma}_0}}{{({\frac{{{r^2} + {L^2} - 2rL \cos\theta }}{{{r_0}}} + 2(r\cos\theta - L)}) d}}\left. {\arctan\left( {\frac{{{\left( {4{r_0} + (4r \cos\theta - 4L) + \frac{{{r^2} + {L^2} - 2rL \cos\theta}}{r_0}}\right)}{\sqrt{\frac{L}{2{r_0}-L}}} - 2r \sin\theta}}{{({\frac{{{r^2} + {L^2} - 2rL \cos\theta }}{{{r_0}}} + 2(r \cos\theta - L)} )}}} \right)} \right|_{-\frac{M}{2}\varepsilon}^{\frac{M}{2}\varepsilon} 
			  \buildrel (b) \over = \frac{{{\bar{\gamma}_0 \pi}}}{{d(r \cos\theta - L)}}.
			\end{aligned} 
		  \end{equation}
		\endgroup
		\hrulefill
	\end{figure*}

	\section{Proff of Corollary 1}

	As shown in Fig.~\ref{System_model}, when the arc support $L \rightarrow 0$, the arc radius $r_0 \rightarrow \infty$, and the central angle $\alpha \rightarrow 0$.
	The resulting SNR in~\eqref{max_SNR} reduces to 
	\begingroup\makeatletter\def\f@size{7.5}\check@mathfonts
		\begin{equation}
			\begin{aligned}
				&\gamma_{\mathrm{UAA}} = \frac{\bar{\gamma}_0 \left( {M-1} \right)}{\alpha}\frac{2}{{{r^2} + 2r{r_0} \cos\theta}} \\
				&\quad \quad \times \left. \arctan\left( {\frac{{\left( {r^2} + 4r_0^2 + 4r{r_0} \cos\theta \right) \tan{\frac{\alpha}{4}} - 2r{r_0} \sin\theta }}{{r^2} + 2r{r_0} \cos\theta }} \right) \right|_{-\frac{M}{2}\varepsilon }^{\frac{M}{2}\varepsilon} \\
				\buildrel (a) \over =& \frac{\bar{\gamma}_0 r_0}{d}\frac{2}{{{r^2} + 2r{r_0} \cos\theta}} \\
				& \quad \times \left. \arctan\left( {\frac{{\left( {r^2} + 4r_0^2 + 4r{r_0} \cos\theta \right) \frac{(M-1)d}{4 r_0} - 2r{r_0} \sin\theta }}{{r^2} + 2r{r_0} \cos\theta }} \right) \right|_{-\frac{M}{2}\varepsilon }^{\frac{M}{2}\varepsilon}\\ 
				=& \mathop {\lim}\limits_{{r_0} \to \infty} \frac{\bar{\gamma}_0}{d} \frac{2}{\frac{r^{2}}{r_{0}}+2 r \cos\theta} \\
				& \quad \quad \quad \times \left. \arctan\left(\frac{\frac{4r_{0}^{2} + 4 r r_{0} \cos\theta + r^{2}}{4r_{0}^{2}} M d-2 r \sin\theta}{\frac{r^{2}}{r_{0}}+2 r \cos\theta}\right) \right|_{-\frac{M}{2}\varepsilon}^{\frac{M}{2}\varepsilon} \\
				\buildrel (b) \over =& \frac{\bar{\gamma}_0}{d r \cos\theta}\left[ \arctan \left(\frac{M d}{2 r \cos\theta}-\tan \theta\right) +\arctan \left(\frac{M d}{2 r \cos\theta}+ \tan\theta \right)\right],
			\end{aligned} 
		\end{equation}
	\endgroup
	where $(a)$ holds due to $\sin \frac{\alpha_0}{2} = \frac{d}{2 r_0}$ and $\sin x \approx \tan x \approx x$ at $x \rightarrow 0$, and $(b)$ holds due to $\frac{a_1 x^2 + b_1 x +c_1}{a_2 x^2 + b_2 x +c_2} = \frac{a_1}{a_2}$ at $x \rightarrow +\infty$. Thus, the proof of Corollary 1 is completed.

	\section{Proff of Corollary 2}

	As shown in Fig.~\ref{System_model}, when the arc support $L$ remains constant and the number of antennas $M \rightarrow \infty$, the arc radius $r_0 \rightarrow \infty$. Then, the resulting SNR in~\eqref{max_SNR} reduces to~\eqref{M_inf}, as shown on top of this page, where $(a)$ holds due to $\sin \frac{\alpha_0}{2} = \frac{d}{2 r_0}$ and $\sin x \approx \tan x \approx x$ at $x \rightarrow 0$, and $(b)$ holds due to $\arctan x = \frac{\pi}{2}$ at $x \rightarrow +\infty$. Thus, the proof of Corollary 2 is completed.

	\bibliographystyle{IEEEtran}
	\bibliography{IEEEabrv,mybib}

\begin{thebibliography}{10}
\providecommand{\url}[1]{#1}
\csname url@samestyle\endcsname
\providecommand{\newblock}{\relax}
\providecommand{\bibinfo}[2]{#2}
\providecommand{\BIBentrySTDinterwordspacing}{\spaceskip=0pt\relax}
\providecommand{\BIBentryALTinterwordstretchfactor}{4}
\providecommand{\BIBentryALTinterwordspacing}{\spaceskip=\fontdimen2\font plus
\BIBentryALTinterwordstretchfactor\fontdimen3\font minus \fontdimen4\font\relax}
\providecommand{\BIBforeignlanguage}[2]{{%
\expandafter\ifx\csname l@#1\endcsname\relax
\typeout{** WARNING: IEEEtran.bst: No hyphenation pattern has been}%
\typeout{** loaded for the language `#1'. Using the pattern for}%
\typeout{** the default language instead.}%
\else
\language=\csname l@#1\endcsname
\fi
#2}}
\providecommand{\BIBdecl}{\relax}
\BIBdecl

\bibitem{you2024next}
C.~You, Y.~Cai, Y.~Liu, M.~Di~Renzo, T.~M. Duman, A.~Yener, and A.~L. Swindlehurst, ``Next generation advanced transceiver technologies for {6G} and beyond,'' \emph{arXiv preprint arXiv:2403.16458}, 2024.

\bibitem{10496996}
H.~Lu, Y.~Zeng, C.~You, Y.~Han, J.~Zhang, Z.~Wang, Z.~Dong, S.~Jin, C.-X. Wang, T.~Jiang, X.~You, and R.~Zhang, ``A tutorial on near-field {XL-MIMO} communications towards {6G},'' \emph{IEEE Commun. Surveys Tuts.}, Apr. 2024, {Early Access}.

\bibitem{10239282}
C.~Wu, C.~You, Y.~Liu, L.~Chen, and S.~Shi, ``Two-stage hierarchical beam training for near-field communications,'' \emph{IEEE Trans. Veh. Technol.}, vol.~73, no.~2, pp. 2032--2044, Feb. 2024.

\bibitem{9913211}
Y.~Zhang, X.~Wu, and C.~You, ``Fast near-field beam training for extremely large-scale array,'' \emph{IEEE Wireless Commun. Letters}, vol.~11, no.~12, pp. 2625--2629, Dec. 2022.

\bibitem{zhou2024multi}
C.~Zhou, C.~You, Z.~Huang, S.~Shi, Y.~Gong, C.-B. Chae, and K.~Huang, ``Multi-beam training for near-field communications in high-frequency bands,'' \emph{arXiv preprint arXiv:2406.14931}, 2024.

\bibitem{lu2021does}
H.~Lu and Y.~Zeng, ``How does performance scale with antenna number for extremely large-scale {MIMO}?'' in \emph{Proc. IEEE Int. Conf. Commun. (ICC)}, Montreal, QC, Canada, Jun. 2021, pp. 1--6.

\bibitem{10243590}
Z.~Wu, M.~Cui, and L.~Dai, ``Enabling more users to benefit from near-field communications: From linear to circular array,'' \emph{IEEE Trans. Wireless Commun.}, vol.~23, no.~4, pp. 3735--3748, Apr. 2024.

\bibitem{10485481}
X.~Wang, C.~Zhang, H.~Zhang, Y.~Xu, and F.-C. Zheng, ``Near-field codebook design for extremely large cylindrical antenna array systems,'' \emph{IEEE Trans. Commun.}, 2024, {Early Access}.

\bibitem{10005200}
Y.~Xie, B.~Ning, L.~Li, and Z.~Chen, ``Near-field beam training in {THz} communications: The merits of uniform circular array,'' \emph{IEEE Wireless Commun. Lett.}, vol.~12, no.~4, pp. 575--579, Apr, 2023.

\bibitem{chen2024near}
Y.~Chen and L.~Dai, ``Near-field wideband beam training for {ELAA} with uniform circular array,'' \emph{Sci. China Inf. Sci.}, vol.~67, no.~6, pp. 1--14, May. 2024.

\bibitem{lu2021communicating}
H.~Lu and Y.~Zeng, ``Communicating with extremely large-scale array/surface: Unified modeling and performance analysis,'' \emph{IEEE Trans. Wireless Commun.}, vol.~21, no.~6, pp. 4039--4053, Jun. 2021.

\bibitem{gradshteyn2014table}
I.~S. Gradshteyn and I.~M. Ryzhik, \emph{Table of integrals, series, and products}, 7th~ed.\hskip 1em plus 0.5em minus 0.4em\relax New York, NY, USA: Academic, 2007.

\end{thebibliography}
	
\end{document}